\documentclass[journal]{IEEEtran}

\usepackage{mathrsfs}
\usepackage{enumerate}
\usepackage{amsmath}
\usepackage{amsthm}
\usepackage{algorithm,algpseudocode}
\usepackage{caption}
\usepackage{graphicx}
\usepackage{float}
\usepackage{bm}
\usepackage{algorithm}
\usepackage{algpseudocode}
\usepackage{lipsum}
   \usepackage{subcaption}
\usepackage{amssymb}
\usepackage{epstopdf}

\newtheorem{example}{Example}
\newtheorem{theorem}{Theorem}
\newtheorem*{proof*}{Proof}
\newtheorem{proposition}{Proposition}
\newtheorem{lemma}{Lemma}

\begin{document}
\title{Electrical Vehicle Charging Station Profit Maximization: Admission, Pricing, and Online Scheduling}

\author{\IEEEauthorblockN{Shuoyao Wang}, \and \IEEEauthorblockN{Suzhi Bi, }~\IEEEmembership{Member,~IEEE, } \and \IEEEauthorblockN{Ying Jun (Angela) Zhang, }~\IEEEmembership{Senior Member,~IEEE, }\and  \IEEEauthorblockN{and Jianwei Huang}, \IEEEmembership{Fellow,~IEEE}
\thanks{This work was supported in part by the National Basic Research Program (973 program Program number 2013CB336701), and three grants from the Research Grants Council of Hong Kong under General Research Funding (Project number 2150828 and 2150876) and Theme-Based Research Scheme (Project number T23-407/13-N). }

\thanks{S.~Wang, Y.J.~Zhang, and J.W.~Huang are with the Department of Information Engineering, The Chinese University of Hong Kong (E-mail: \{ws013, yjzhang, jwhuang\}@ie.cuhk.edu.hk@ie.cuhk.edu.hk). Y.J.~Zhang  is also with Shenzhen Research Institute, The Chinese University of Hong Kong.}

\thanks{S.~Bi is with the College of Information Engineering, Shenzhen University, Shenzhen, Guangdong, China (E-mail: bsz@szu.edu.cn).}
\vspace{-1cm}}


\maketitle

\begin{abstract}
The rapid emergence of electric vehicles (EVs) demands an advanced infrastructure of publicly accessible charging stations that provide efficient charging services. 
In this paper, we propose a new charging station operation mechanism, the JoAP, which jointly optimizes the EV admission control, pricing, and charging scheduling to maximize the charging station's profit. More specifically, by introducing a tandem queueing network model, we analytically characterize the average charging station profit as a function of the admission control and pricing policies. Based on the analysis, we characterize the optimal JoAP algorithm.  Through extensive simulations, we demonstrate that the proposed JoAP algorithm on average can achieve 330\% and 531\% higher profit than a widely adopted benchmark method under two representative waiting-time penalty rates.      
\end{abstract}

\IEEEpeerreviewmaketitle
\section{Introduction}
%
%
%
%
Environmental awareness and the rising fuel cost have stimulated an increasing interest in electrical vehicles (EVs). Establishing a conveniently available public charging infrastructure is essential  to ensure a large market penetration of EVs \cite{7390308}. Currently, however, the operation of  charging infrastructure is often not very profitable  due to the low expected revenues, high capital expenditures, and high operating and maintenance costs \cite{schroeder2012economics}. 
   
In light of this, several recent studies focused on improving the operation efficiency of EV charging stations (e.g., \cite{you2014efficient}\nocite{7500060}\nocite{6887361}\nocite{ghavami2014nonlinear}-\cite{7352372})  by carefully designing the charging scheduling and pricing mechanisms. In particular, You and Yang in \cite{you2014efficient} characterized  an optimal offline charging scheduling scheme, where ``offline" means that the scheduling decision relies on the noncausal information of future EV charging profiles. Tang and Zhang in \cite{7500060} relaxed  the assumption of noncausal information  by utilizing only the statistical distributions, instead of the exact realizations, of future EV charging profiles. In \cite{6887361}, Tang \emph{et al.} designed an online  charging scheduling algorithm that does not require any future information, not even the distribution information.  Ghavami and Kar in \cite{ghavami2014nonlinear} and Yuan \emph{et al.} in \cite{7352372} further proposed charging scheduling and pricing schemes to incentivize EV users to achieve  social optimality (i.e., minimizing the network-wide charging cost or maximizing the total economic surplus). In brief, various pricing schemes have also been proposed to maximize the charging station's profit through time-scale decomposition, peak valley decomposition and Lagrangian relaxation, and dual decomposition \cite{you2014efficient}-\cite{7352372}. 

Most existing studies, e.g, \cite{you2014efficient}-\cite{ghavami2014nonlinear}, assumed that a charging station has unlimited charging power to accommodate an  infinite number of EVs simultaneously. In practice, however, the total charging power is bounded due to the physical and security constraints of the distribution network. Moreover, the number of EVs that a charging station can accommodate is limited by the hardware and space constraints.
As such, the charging waiting time (defined as the time between the arrival time of the EV to the charging station and the time that the EV starts to receive service) is often unavoidable, which negatively impacts the users' experience.   
Hence it is necessary to implement an effective admission control policy to reduce the impact of the excessive charging waiting time  due to random EV arrivals.  

A  commonly-adopted admission control is the  queue-length based admission (QBA)  policy, where a newly arrived EV is admitted as long as the number of EVs waiting to be served at  the station is below  a specific threshold (e.g., the waiting room in the charging station). However, such a policy  performs poorly in many cases, as illustrated in Section~V. 
In contrast, Wei \emph{et al.}~in  \cite{7504120} proposed an admission control scheme, where the admission decision is based on the charging demands of EVs that have already arrived. The unknown future charging demands, however, were not considered in  \cite{7504120}, resulting in poor profit performance in practical scenarios (see Section V for related examples).


In this paper, we propose a novel EV charging station operating mechanism that  jointly optimizes pricing, charging scheduling, and admission control. The proposed algorithm, referred to as JoAP (joint admission control and pricing),  maximizes the average profit of a charging station. Here the profit corresponds to the difference between the revenue and a penalty proportional to the average charging waiting time. 
The waiting time penalty reflects the EV owners' impatience of waiting in the queue for an excessively long time, which  undermines the reputation of the charging station and reduces his long-term profit.
In the JoAP algorithm, each EV user maximizes his surplus by adjusting his charging demand in response to the charging price and the charging station maximizes his profit by choosing the proper admission control, scheduling and pricing  policies. 

The contributions of this paper are summarized  as follows: 
\begin{enumerate}
 \item  \emph{Admission control, scheduling, and pricing scheme:} To the best of our knowledge, this is the first paper that jointly optimizes pricing, scheduling, and admission control of an EV charging station. In particular,  we propose a novel multi-sub-process based admission control scheme, which allows us to flexibly tradeoff  between the revenue of the charging station and the waiting time of the EVs.

 \item \emph{Tandem queueing model:} We propose a tandem queueing model to analytically capture  the performance  of the proposed JoAP algorithm. More specifically, we obtain closed-form expressions of the average waiting time and admission probability as functions of the chosen algorithm  parameters.    

\item \emph{Optimization of algorithm parameters:} Based on the analysis of the tandem queue, we propose a low-complexity algorithm to compute the close-to-optimal parameters of the JoAP algorithm. Our simulations show that JoAP algorithm on average can achieve 330\% and 531\% higher profit than a widely adopted benchmark method under two representative waiting-time penalty rates.
\end{enumerate}

The rest of this paper is organized as follows. In Section II, we introduce the system model and formulate the problem. In Section III, we analyze the impact of the admission control policy on the admission probability and the average waiting time. In Section IV, we propose an efficient algorithm to simultaneously  maximize the charging station's profit and individual EV user's payoff surplus. Simulation results are presented in Section V. Finally, we conclude this paper in Section VI.
\section{System model}
\subsection{Charging Station Operation}
\vspace{-0.5cm}\begin{figure}[htp]   
	\includegraphics[width=0.45\textwidth]{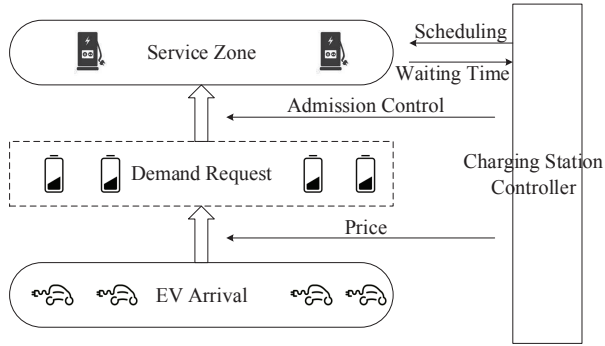}
	\caption{The proposed charging station interaction system \normalsize}
	\label{fig:system}
\end{figure}

We consider a charging station with $m$ charging ports and a sufficiently large number of parking lots (i.e., much larger than $m$), as shown in Fig.~\ref{fig:system}. In this case,  although a large number of EVs can be admitted to  the charging station, at most $m$ of them can be charged (served)simultaneously because of the physical constraints of the power distribution network and safety concerns.   The charging ports are connected to the parking lots through a switch scheduler, which  allows real-time communications and controls   between a particular charging port and a scheduled EV. For the simplicity of analysis, we assume that the cost of connecting  EVs with charging ports is negligible. All charging ports operate  with the same fixed charging power $\alpha$.

The charging station announces  a charging price of $r$ per unit energy to all arriving EVs. An EV $i$'s payment  payment to the charging station is the product of $r$ and the EV's demand $d_i$.  A long waiting time negatively affects the EV users' experience, which may lead to customer churn  in the long run. Thus, the charging station aims to determine the optimal pricing and admission control policy to maximize his average profit, which is the revenue minus the penalty due to EVs' waiting. 

EVs arrive at the charging station according to a Poisson random process \cite{7504120}, and each EV expects the charging station to fulfill his demand as soon as possible.  When an EV $i$ arrives, it attempts to maximize his  surplus by choosing his charging demand $d_i$ according to the charging price $r$. Based on the requested demand $d_i$,  the charging station decides whether to admit the EV.   The charging station will optimize his  an admission control policy to avoid excessive delay of admitted EVs. Once admitted, the EVs are charged on a first come first serve (FIFO) basis to fulfill their charging demands. It has been shown in \cite{erickson2010tardiness} that, when all EVs are homogeneous,  the FIFO  policy is equivalent to the shortest job first policy, and therefore is optimal in terms of minimizing the average waiting time. 
   
\subsection{Optimization from EVs' Perspective }
For simplicity, we consider homogeneous EVs \cite{5717547}. More specifically, all EVs have the same battery capacity $\varphi$ and the same utility function. Without loss of generality, we use the utility function $U(d)$ proposed in \cite{khan2011user} as an example to conduct the formulations. Notice that all analytical results still hold for any general increasing concave utility function.  Consequently, $U(\varphi)$ is the maximum utility that an EV can receive. An EV $i$ determines his charging demand to maximize his consumer surplus (i.e., utility minus payment),

\small\vspace{-0.5cm}\begin{subequations} 
\begin{align} 
  \underset{d_i}{\text{max}} &~~ U(d_i) -rd_i\\
\text{s.t.} &~~ 0\leq d_i \leq \varphi. 
\end{align}
\end{subequations}\normalsize
In particular, we consider the following concave utility function \cite{khan2011user},  where $\beta$ is the elasticity parameter, 
 
 \small\vspace{-0.1cm}\begin{equation} 
U(d)=U(\varphi)\frac{1-e^{-\beta d}}{1-e^{-\beta \varphi}}, \quad \forall 0\leq d \leq \varphi.
\end{equation}
\normalsize
As Problem (1) is a concave maximization problem, we can compute the optimal demand $d^*$ as a function of the service price $r$ as follows,

\small\begin{equation}
d^{*}(r) = \begin{cases} -\frac{\ln(\frac{1-e^{-\beta\varphi}}{U(\varphi)\beta} r)}{\beta}, &\mbox{if } r \leq \frac{U(\varphi)\beta}{1-e^-\beta\varphi}, \\ 
0, & \mbox{otherwise}. \end{cases} 
\end{equation}\normalsize
We can show that, $d^{*}(r)$ is a decreasing function of the price $r$ announced by the charging station, and becomes 0 when $r$ is too high. Note that $d^*(r)$ is the same for all EVs, since the EVs are homogeneous.

In this paper, we assume that the charging station knows the homogeneous utility function (2). Accordingly, the station can predict EV's demand  $d^{*}(r)$ in response to the price $r$ as in (3). As the user demand $d^{*}(r)$ has a one-to-one correspondence with  price $r$. Thus, optimizing $r$ is equivalent to optimizing $d$ in the rest of the paper.  

\subsection{Optimization from  Charging Station's Perspective}
 
Let  $\mathcal{V}$ denote the set of all EVs that arrive at the parking station during the time period of interest (e.g., 4 hours in our simulations). For each EV  $i \in \mathcal{V}$, the charging station makes a binary admission decision $x_i^{\mathcal{V}}(\pi_n,d)$, where $x_i^{\mathcal{V}}(\pi_n,d)=1$ if EV $i$ is admitted, and $x_i^{\mathcal{V}}(\pi_n,d)=0$ otherwise.  Here, $\pi_n$ denotes an admission policy, which will be detailed in Section III.A.  Consequently, the average admission probability is  $P_{\pi_n}(d)=\mathrm{E}_{\mathcal{V}}\left[\frac{1}{|\mathcal{V}|}\sum_{i \in \mathcal{V}}x_i^{\mathcal{V}}(\pi_n,d) \right]$, where $|\mathcal{V}|$ denotes the the cardinality of $\mathcal{V}$. Moreover, the average waiting time achieved under a policy   $\pi_n$ is a function of the demand $d$ and the EV arrival process $\mathcal{V}$,  denoted as  $\omega _{\pi_n} (\mathcal{V},d)$. Accordingly, the waiting time averaged over all the possible EV arrivals is denoted by $\omega_{\pi_n}(d)  \triangleq E_{\mathcal{V}}[\omega_{\pi_n}(\mathcal{V},d)]$.

By satisfying an EV's  charging demand $d$, the charging station receives a payment of $rd$, and pays an  electricity cost of $p_ed$ to the utility company, where $p_e$ is the  electricity price. The penalty related to the average waiting time is denoted by $h(\omega_{\pi_n}(d))$, where $h(\omega)$ is a general non-decreasing convex function of $\omega$ \cite{li2011delay}.   
Based on this, we formulate the charging station’s profit-maximization problem as in Problem (4).\footnote{Problem (4) doesn't consider the penalty of denying the EVs. However, we can consider this by simply adding a linear term of $P_{\pi_n}(d)$ to the objective function. Doing so does not affect the structure of the problem, and our analysis will remain unchanged. For simplicity of exposition, this linear term is omitted for the time being.}

\small\vspace{-0.5cm} 
\begin{subequations}
\begin{align}
\underset{\pi_n,d}{\text{max}} &~~ P_{\pi_n}(d) \left(r-p_e\right) d-h \left(\omega_{\pi_n}(d)\right) \\
\text{s.t.} &~~ d \geq 0, \quad i \in \mathcal{V}, \\
&~~ \pi_n \in \Pi, \\
&~~ d= -\frac{1}{\beta}\ln\left(\frac{1-e^{-\beta\varphi}}{U(\varphi)\beta} r\right),
\end{align}
\end{subequations}\normalsize
where the feasible set $\Pi$ will be introduced in Section III.A. The detailed expressions of  $P_{\pi_n}(d)$ and $\omega_{\pi_n}(d)$ will be given  in Section III.B and Section III.C, respectively.

\section{Multi-Sub-Process Admission and Queueing Analysis}
In this section, we first propose a multi-sub-process admission control scheme. Then, we present a tandem queueing model to analyze the impact of admission control policy and pricing  decision on $P_{\pi_n}(d)$ and $\omega_{\pi_n}(d)$. 
\subsection{Admission Control and Queueing Model}
 The objective of  admission control is to admit a large number of users with a guaranteed QoS. Let us first consider an extreme case of complete arrival process regulation, i.e., the inter-arrival time of two successively admitted EVs is always larger than a predefined threshold as the result of the admission control. If such a threshold is large enough, then the waiting time of every admitted EV will be  zero \cite{erickson2010tardiness}. However, under this overly conservative admission control policy, the charging station utilization can be very low, hence not achieving the maximum profit. To achieve a good balance among the waiting time, admission rate, and server utilization, we propose a multi-sub-process admission control scheme consisting of $n$ sub-processes. In particular, the inter-arrival time of two consecutively admitted EVs of the same sub-process must be larger than a threshold, denoted as $T_v$. An EV is admitted as long as it can fit in one of the sub-processes.   With some  abuse of notations, we use $\pi_n$ to denote the proposed admission control policy involving  $n$ sub-processes. Hence, the feasible set of all admission control policies considered in this paper is $\Pi=\{\pi_n | n \in \mathcal{N}^{+}\}$.
\begin{figure}[htp]    
	 \vspace{-0.3cm}\includegraphics[width=0.40\textwidth]{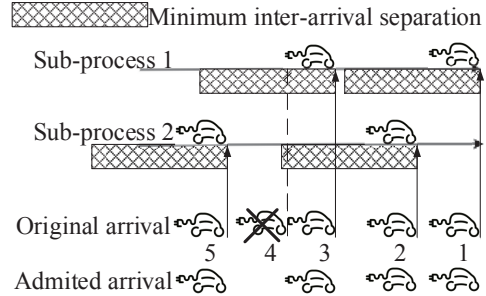}
	\caption{\small Admission control example illustrated in Example~1 \normalsize}
	\label{fig:example}
\end{figure}
 \vspace{-0.4cm}\begin{example}
Consider a $\pi_2$-admission policy that consists of two sub-processes, both having the same  minimum inter-arrival time  $T_v$, as shown in Fig. 2. When EV $1$ arrives, we assign it to sub-process $1$. When EV $2$ arrives, we cannot assign to  sub-process $1$ as the inter-arrival time between EV 1 and EV 2 is shorter than $T_v$ (the length of shadowed rectangle). Hence, we assign EV 2 to  sub-process $2$. For EV $3$, we can assign it  to sub-process $1$. However, when EV $4$ arrives, both sub-processes are \lq\lq occupied\rq\rq. Therefore, EV $4$ has to be rejected.  When EV $5$ arrives, sub-process $2$ becomes available again  (due to the large enough inter-arrival time between EV 2 and EV 5). Hence, we  accept EV $5$ and assign it to sub-process 2. 
\end{example}
We would like to emphasize that the multi-sub-process scheme can represent a wide range of admission control policies. On one hand, the admitted traffic is completely regulated if there is only one sub-process, i.e., $n=1$. On the other hand, when $n$ approaches infinity, all EVs will be admitted regardless of the underlying distribution of the arrival process. Thus, choosing proper values of $n$ and $T_v $ allows us to balance the trade-off between the waiting time and admission rate, and eventually maximizes the charging station profit.

The admission process governed by the multi-sub-process scheme can be modeled as a virtual queueing system with zero buffer and $n$ servers, as represented by  Q.1 (an $M$/$T_v$/$n$/$n$ queue\footnote{We can represent a single queue using Kendall's notation in the form $A$/$S$/$C$/$C$+$K$, where $A$ describes the inter-arrival times, $S$ describes the service time, $C$ describes the number of servers, and $C$+$K$ describes the number of spaces in the system. When the $K$ parameter is not specified (e.g. $M$/$M$/$1$ queue), it is assumed that  $K =\infty$.    In Kendall's notation: $M$ stands for Markov or memoryless process, $D$ stands for deterministic process, $G$ stands for general and corresponds to  an arbitrary probability distribution, $Ph$ stands for phase-type process (the process that constructed by a convolution or mixture of exponential process), and $\cdot$ stands for any process.}) in Fig.~3. Each virtual server corresponds to a sub-process, which  has a deterministic service time $T_v$. The arrival of Q.1 is the EV arrival process $\mathcal{V}$. As the buffer is zero for Q.1,  an EV will be declined for service if it finds all virtual servers are busy (i.e., all sub-processes are occupied) upon arrival.   Otherwise, the EV is admitted and will occupy an idle virtual server for a fixed time period of $T_v$. The departure from Q.1 means that the EV is  admitted to the charging station.

 Once EVs are admitted, they are served in the charging station according to the FIFO policy. We model  the queueing system in the charging station as Q.2 in Fig. 3, where the $m$ charging ports represent $m$ servers, each with a deterministic service time $d/\alpha$ , where $d$ and $\alpha$ are the charging demand per EV and the fixed common charging rate per charging port, respectively. Note that the departure process  of Q.1 is the arrival process  of Q.2. To ensure the stability of  Q.2, the  inter-departure time of Q.1 must be greater than the average service time of Q.2, i.e., $nT_v > md / \alpha$. We can equivalently  represent this constraint as  $n T_v =\tau m d / \alpha$, where $\tau > 1$.To sum it up, the determination of an admission control policy involves two decision variables,  $\tau$ and $n$, with which can compute  $T_v =\frac{\tau m d}{n\alpha}$.  In the following, we will consider optimize  $n$ under a fixed value of $\tau$.  Without loss of generality, we assume that $\alpha$ equals $1$. We will examine the impact of $\tau$ in Section~V.

 In practice,  a well regulated arrival process seldom yields a long queue length \cite{van2006performance}. Consequently,   we ignore the impact of buffer of Q.2 and assume that it is infinite in the following analysis.  We can going to propose a  JoAP algorithm  that optimize the performance of  an $\underbrace{M\text{/}T_v\text{/}n\text{/}n}_\text{Q.1}$\,+\,$\underbrace{\cdot\text{/}d\text{/}m}_\text{Q.2}$ tandem queueing network. 

Before concluding  this subsection, we would like to emphasize that Q.1 in Fig. 3 is a virtual queue that does not exist in reality. We consider Q.1 for the purpose of  analyzing the admission control policy. Queue Q.2 is a real queue corresponding to the service in the charging station. As such, the admission probability is the probability that a new arrival is admitted to Q.1, and the charging waiting time is the waiting time in Q.2. 
\begin{figure}[htp]   
\hspace{-0.1cm}\includegraphics[width=0.5\textwidth]{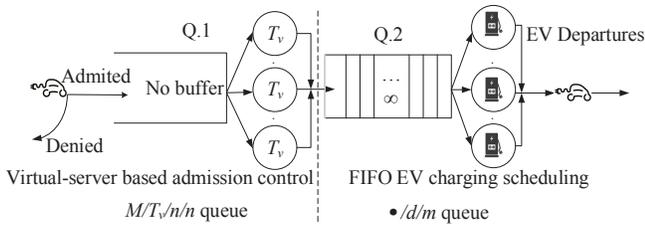}
	\caption{\small The tandem queueing network model \normalsize}
	\label{fig:queue}
\end{figure}

\subsection{Admission Probability}
Previous queueing literature (e.g., \cite{allen2014probability}) have numerically analyzed the performance of  $M$/$D$/$C$/$C$+$K$ queues (e.g., Q.1 in Fig.~3) without analytical characterization of the system performance. H.~Tijms in \cite{tijms2008note} showed  that a two-phase process server can be used to approximate a deterministic server with a marginal performance gap. 
Based on this approximate model, we derive a closed-form expression of steady-state probabilities of Q.1 in the following Lemma~1.
To the best of our acknowledgment, this paper is the first analytical study of the $M$/$D$/$C$/$C$+$K$ system with $K=0$ (i.e., zero buffer).

\begin{lemma}
Consider an $M$/$D$/$n$/$n$ queue with a Poisson process with a arrival rate $\lambda$, a deterministic service time $\tau md / n $, and zero buffer-size. The steady-state probability of state $i$ (i.e., the probability that the system has $i$ users  being served simultaneously) can be  calculated based on the two-phase-process approximation in \cite{tijms2008note} as follows,

\small\begin{equation}
P_i(n,d)=\frac{ (\frac{d\tau m\lambda }{n }) ^i}{i!\sum_{j=0}^n \frac{(\frac{d\tau m\lambda }{n }) ^j}{j!}}.
\end{equation}\normalsize

The admission probability of Q.1 is:
\small\begin{equation}
P_{\pi_n}(d) = 1-P_n(n,d)=1-\frac{  \left(\frac{ d\tau m\lambda}{n}\right)^n e^{-\frac{\tau m d \lambda}{n }}}{\Gamma \left(n+1,\frac{\tau m d \lambda}{n }\right)}.
\end{equation}
\end{lemma}\normalsize

We can prove Lemma~1 by induction, with the detailed proof in the on-line technical report \cite{website2013} due to the page limit. The validity of  Lemma~1 is  verified in Fig.~\ref{fig:adm},
\begin{figure}[htp]    
	\includegraphics[width=0.5\textwidth]{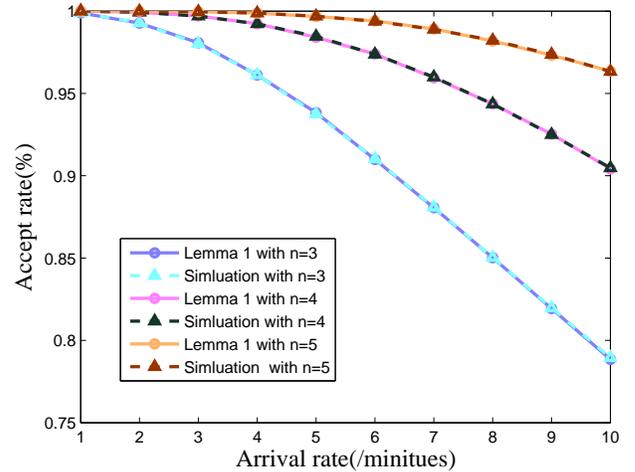}
	\caption{\small The comparison of the admission probability  between the simulation and the approximation in Lemma 1, with $\tau=1.01$, $m=4$, $\beta=0.05$, $\alpha=3.3kW$, $d=\varphi$, and $\gamma=35kWh$ \normalsize}
	\label{fig:adm}
\end{figure}
where we compare the admission rate derived in (6) with the simulation results (without any approximation). We choose the number of  servers in Q.1, $n$, to be  $3$, $4$, and  $5$, respectively. Each point corresponds to the average over 1000 time periods.  The maximum gap between the analysis and simulation is  0.01$\%$, which verifies the accuracy of the results in Lemma 1.

\subsection{Average Waiting Time}

\subsubsection{Admitted-arrival}
 To study the average waiting time in Q.2, we derive the PDF (probability density function) of the inter-arrival time of Q.2. 
\begin{lemma}
The PDF of the inter-arrival time of admitted arrivals of Q.2 is
\small\begin{equation}
f_{X}(x)=\begin{cases}
    \sum_{i=0}^n\frac{i}{T_v}\left(\frac{T_v-x}{T_v}\right)^{i-1} P_i(n,d),& \text{if } x\leq T_v,\\
    0,              & \text{otherwise.}
\end{cases}
\end{equation}
\end{lemma}\normalsize   
\begin{proof}
Recall that the arrival process of Q.2 is the departure process of Q.1. According to \cite{sztrik2012basic}, the residual service time of a queueing system is the service time remaining to a job under service when the system is observed at any time. The residual service time of Q.1  follows a uniform distribution in $[0,T_v]$, as the arrival process is memory-less (Poisson) and the buffer size is zero \cite{tijms2006new}. When Q.1 is at a particular state $i$, the probability of no departure during the next period of  time of a length $x$ is equal to the probability that  the residual service times of all existing jobs are no-less than $x$, i.e.,  $\left(\left(T_v-x\right) / T_v\right)^i$. Consequently, the probability of the first departure time (after the observation time point) being no greater than $x$ is  $1-\left(\left(T_v-x\right) / T_v\right)^i$. Therefore, the CDF of the inter-departure time of Q.1 (i.e., the inter-arrival time of Q.2), denoted by  $X$, is,

\small\vspace{-0.1cm}\begin{equation}
F_{X}(x)=\begin{cases}
    \sum_{i=0}^n\left(1-\left(\frac{T_v-x}{T_v}\right)^i\right)P_i(n,d),& \text{if } x\leq T_v,\\
    0,              & \text{otherwise.}
\end{cases}
\end{equation}\normalsize

Taking the derivative of (8)  yields the PDF in Lemma~2.
\end{proof}

\subsubsection{Phase-type Approximation}
We now derive the average waiting time of Q.2 with the phase-type approximation. So far, there does not exist a general closed-form expression for the  waiting time distribution of a $GI$/$D$/$m$ queue (e.g., Q.2 in Fig.~3) \cite{tijms2003first}, where $GI$ means a general arrival process. To overcome this difficulty, \cite{tijms2003first} showed  that the waiting time distribution of a $GI$/$D$/$m$ queue is the same as that of a $GI^{(m^{*})}$/$D$/$1$  queue, where $GI^{(m^{*})}$  denotes a coordinated inter-arrival time  process that is distributed as the sum of $m$ inter-arrival times of a $GI$/$D$/$m$ queue. Let  $Y$ denote the coordinated inter-arrival time of the $GI^{(m^{*})}$/$D$/$1$. The mean and variance of $X$ and $Y$ are related by $\mu_Y =mE(X)$ and $\sigma_Y^2=mE(X)^2-m\left(\left(X\right)\right)^2$.

Furthermore, a $GI^{(m^{*})}$/$D$/$1$ queue can be approximated by a $Ph$/$D$/$1$  queue, where $Ph$ means the phase-type process \cite{tijms2003first}. One of the most widely used phase-type distribution is the mixture exponential distribution, which is defined as the mixture of two exponential distributions with means $1/\lambda_1 $and $1/\lambda_2$, and weights $\gamma$ and $1-\gamma$, respectively.  Specifically, the PDF is given by 

\small\vspace{-0.1cm}\begin{equation}
	f_{\text{Ph}}(x)=\gamma e^{-\lambda_1 x}+(1-\gamma) e^{-\lambda_2 x}.
\end{equation}\normalsize
In this paper, we replace the inter-arrival distribution of Q.2 with the mixture exponential distribution in (9). To ensure that the first and second moments of the mixture exponential distribution are equal to those of $Y$, we set $
\frac{1}{\lambda _1}+\frac{1}{\lambda_2} =2\mu_Y$, $\frac{1}{\lambda _1^2}+\frac{1}{\lambda _2^2} =\sigma_Y^2$, and $\gamma=\frac{1}{2}$. In this way, we can approximate the waiting time distribution of Q.2 by that of the $Ph$/$D$/$1$  queue.

 Let $\rho=\lambda P_{\pi_n}(d)d /\left(m\right)$ denote the load density admitted to the charging station.  We derive in the following Theorem 1  the approximated average waiting time of the charging station.

\begin{theorem}
 The approximated average waiting time at the charging station for the admitted EVs is

\small\vspace{-0.1cm}\begin{equation}
\omega_{\pi_n}(d) =\frac{\rho d}{2(1-\rho)}\left[d^2+2d\mu_Y+\sigma^2_Y\right].
\end{equation} \normalsize
Moreover, $\omega_{\pi_n}(d)$ is an  increasing convex function in $d$ for a fixed $n$.
\end{theorem}

The proof of  Theorem~1 can be find in the Appendix~B of the online technical report \cite{website2013}. Let us verify the approximation by comparing the average waiting time in (10) with  simulation results (without any approximation). In Fig.~\ref{fig:wt}, for each pair of arrival rate and individual demand, we simulate 1000 independent 1000-hour arrival processes $\mathcal{V}$ and plot the average admission rates.  The difference is no more than 0.1\%.  
\begin{figure}[htp]    
	\vspace{-0.5cm}\includegraphics[width=0.5\textwidth]{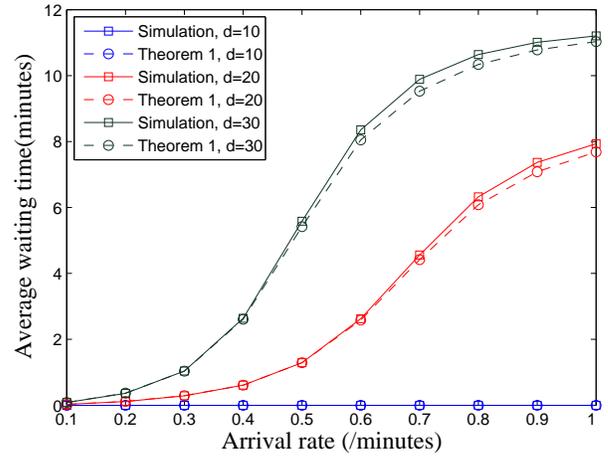}
	\caption{\small The comparison of the average waiting time between the simulation and the approximation in Lemma 1, with $m = 4$.  \normalsize}
	\label{fig:wt}
\end{figure}

\section{Optimization Problem Recasting and Profit Maximization}
\subsection{Optimization Problem Recasting}
With the tandem queueing analysis, we can rewrite (4) as

\small\vspace{-0.3cm}\begin{subequations}
\begin{align}
\underset{n,d}{\text{max}}&~~ s(n,d)=P_{\pi_n}(d)\left(\frac{d e^{-\beta d}}{\xi}-d p_e \right)-h\left(\omega_{\pi_n}\left(d\right)\right)\\
\text{s.t.}&~~ 0\leq d\leq \varphi, \quad n \in \mathcal{N}^{+}, 
\end{align}
\end{subequations}\normalsize
where $\xi=\frac{1-e^{-\beta\varphi}}{U(\varphi)\beta}$, $P_{\pi_n}(d)$ and $\omega_{\pi_n}\left(d\right)$ are given in (6) and (10), respectively. To solve the integer programming Problem (11) efficiently, we replace the decision variable $n$ with $P \triangleq P_{\pi_n}(d)$. This is because for a particular feasible  $(P,d)$, we can find a unique $n$ that satisfies equation (6) and the objective function (12a) is concave under the conditions in (12b). Accordingly,  (11) can be equivalently expressed as 

\small\vspace{-0.1cm}\begin{subequations}
\begin{align}
\underset{P,d}{\text{max}}&~~ \hat{s}(P,d)=P\left(\frac{d e^{-\beta d}}{\xi}-d p_e \right)-h\left(\omega_{\pi_n}\left(d\right)\right)\\
\text{s.t.}&~~ 0\leq d\leq \varphi, \quad P \in (0,1), \\
&~~ P \in \{P_{\pi_n}(d)|\forall n \in \mathcal{N}^{+}, \forall d \in [0,\varphi] \}.
\end{align}
\end{subequations}\normalsize

 Shaked and  Shanthikumar in \cite{shaked1988stochastic} showed that the average waiting time of a $GI$/$GI$/$1$ queue  with first-come-first-served order is jointly convex in the effective-arrival-rate and the service rate. The effective-arrival-rate of  the corresponding coordinated queue (a $GI^{(m^{*})}$/$D$/$1$ queue) of Q.2 is $\frac{P\lambda}{m}$. By the composition rule, we can see that $-h\left(\omega_{\pi_n}\left(d\right)\right)$ is jointly concave in $(\frac{P\lambda}{m},d)$ (thus in $(P,d)$). This, together with the fact that $P\left(\frac{d e^{-\beta d}}{\xi}-d p_e \right)$ is jointly concave in $(P,d)$, implies that (12a) is a jointly concave function in $(P,d)$. If we  ignore the integer constraint in (12c), then Problem (12) can be solved efficiently by the gradient method, with the optimal solution denoted as $(P^{v},d^{v})$. Accordingly, $n^v$ can be obtained by solving (6) given $(P^v, d^v)$. However, $n^v$ obtained  through this approach  does not necessarily satisfy the integer constraint in (11b). In the following Lemma~3, we show that the optimal solution to Problem (11) can be easily obtained by rounding $n^v$ to the nearest integer. In the lemma, we will use the notation of  $d^{*}_{n} = \arg \max_d s(n,d)$. Then, we have the following characterization of  the optimal solution $(n^*,d^*)$ to Problem (11).

\begin{lemma}
Given that $(n^{v},d^{v})$ is an optimal solution to Problem (12a-b)  (without considering the constraint (12c)), then the optimal solution to Problem (11) is either $(\lfloor n^{v} \rfloor,d^{*}_{\lfloor n^{v} \rfloor})$ or  $(\lceil n^{v} \rceil,d^{*}_{\lceil n^{v} \rceil})$, whichever yields the larger objective function value.\footnote{$\lfloor n\rfloor$ and $\lceil n \rceil$ denote the largest integer no greater than $n$ and the smallest integer no less than $n$.}  
\end{lemma}
\begin{proof}
First, we show that for any $\hat{n}<\lfloor n^{v} \rfloor$, $s(\hat{n},d^{*}_{\hat{n}}) \leq s(\lfloor n^{v} \rfloor,d^{*}_{\lfloor n^{v} \rfloor})$. It's equivalent to showing that for any $\hat{n}<\lfloor n^{v} \rfloor$, we can find an $(\lfloor n^{v} \rfloor,d_1)$ such that $s(\lfloor n^{v} \rfloor,d_1) \geq s(\hat{n},d^{*}_{\hat{n}})$. From (6), $P_{\pi_n(d)}$ is monotonically increasing in both $n$ and $d$. Thus, we can always find a point $(P_{\pi_{\lfloor n^{v} \rfloor}}(d_1), d_1)$ in the line segment between  $(P_{\pi_{\hat{n}}}(d_{\hat{n}}^*),d_{\hat{n}}^*)$ and $(P^v,d^v)$. The monotonicity of $P_{\pi_n}(d)$ guarantees the existence and uniqueness  of $(P_{\pi_{\lfloor n^{v} \rfloor}}(d_1),d_1)$. Due to the joint concavity of $\hat{s}$ in $(P,d)$, we have $\hat{s}(P^{v},d^{v})\geq \hat{s}(P_{\pi_{\lfloor n^{v} \rfloor}}(d_1),d_1)\geq \hat{s}(P_{\pi_{\hat{n}}}(d_{\hat{n}}^*),d_{\hat{n}}^*)$. Due to the equivalence between Problem (11) and Problem (12), we have $s(n^{v},d^{v})\geq s(\lfloor n^{v} \rfloor,d_1) \geq s(\hat{n},d_{\hat{n}}^*)$.  Likewise, we can prove that for any $\hat{n}>\lceil n^{v} \rceil$, $s(\hat{n},d^{*}_{\hat{n}}) \leq s(\lceil n^{v} \rceil,d^*_{\lceil n^{v} \rceil})$. Therefore, we can conclude that the optimal solution to Problem (11) is either $s(\lfloor n^{v} \rfloor,d^{*}_{\lfloor n^{v} \rfloor})$ or $s(\lceil n^{v} \rceil,d^{*}_{\lceil n^{v} \rceil})$.
\end{proof}

Lemma~3  indicates that we can obtain the optimal $n^*$ by rounding $n^v$. What remains is how  to calculate $d^{*}_{\lfloor n^{v} \rfloor}$ and $d^{*}_{\lceil n^{v} \rceil}$ efficiently. The following Lemma 4 indicates that  $d^{*}_{\lfloor n^{v} \rfloor}$ and $d^{*}_{\lceil n^{v} \rceil}$ can be easily obtained using single-variable convex optimization methods, e.g., the gradient search method.

\begin{lemma}
If $n$, $s(n,d)$ is concave in $d$ for $d \in \{d|\frac{1}{\xi}\left(e^{-\beta d}-d \beta e^{-\beta d}\right)-p_e \geq 0\}$.   Moreover, $s(n,d)$ is concave in d when $n=n^*$. 
\end{lemma}
\begin{proof} 
We first prove that given any $n$, $s(n,d)$ is concave in $d$ for $d \in  \{d|\frac{1}{\xi}\left(e^{-\beta d}-d \beta e^{-\beta d}\right)-p_e \geq 0\}$.  For any $d$ such that $\frac{1}{\xi}\left(e^{-\beta d}-d \beta e^{-\beta d}\right)-p_e \geq 0$, $\frac{d e^{-\beta d}}{\xi}-d p_e$ is a positive increasing concave function in $d$. Meanwhile, it can be seen from (5) that $P_{\pi_n}(d)$  is a positive decreasing concave function in $d$. Therefore, the product $P_{\pi_n}(d)\left(\frac{d e^{-\beta d}}{\xi}-d p_e \right)$ is concave in $d$.  
According to Theorem~1, we have  $\forall n \geq m$,  $\frac{\partial^2 \omega_{\pi_n}\left(d\right)}{\partial d^2}>0$, and $\frac{\partial \omega_{\pi_n}\left(d\right)}{\partial d}>0$. This, together with the fact that $h(\omega)$ is a non-decreasing convex function, implies that $-h\left(\omega_{\pi_n}\left(d\right)\right)$ is also concave in $d$.  Hence, $s(n,d)$ is concave in $d$ for $d \in \{d|\frac{1}{\xi}\left(e^{-\beta d}-d \beta e^{-\beta d}\right)-p_e \geq 0\}$.

We now prove that $s(n, d)$ is concave in $d$ at an optimal $n^*$. This can be proved by showing that the condition $\frac{1}{\xi}\left(e^{-\beta d}-d \beta e^{-\beta d}\right)-p_e \geq 0$ is satisfied at the optimal solution, which we will show by contradiction.  Suppose that  $\frac{1}{\xi}\left(e^{-\beta d}-d \beta e^{-\beta d}\right)-p_e < 0$ holds for an optimal solution $(n^*,d^*)$. In this case, the objective in (11a) is monotonically decreasing in $d$, because the derivative of the first term in (11a) is negative in the domain and the second term in (11a) monotonically decreases with $d$. This contradicts with the assumption that $(n^*,d^*)$ is an optimal solution. Thus $\frac{1}{\xi}\left(e^{-\beta d}-d \beta e^{-\beta d}\right)-p_e \geq 0$ must hold for an optimal solution to Problem (11).   
\end{proof}

With Lemmas 3 and 4 , we propose  a 3-step optimal solution algorithm to Problem (11) in Fig.~\ref{fig:fc}.

\begin{figure}
\centering
   \includegraphics[width=1\linewidth]{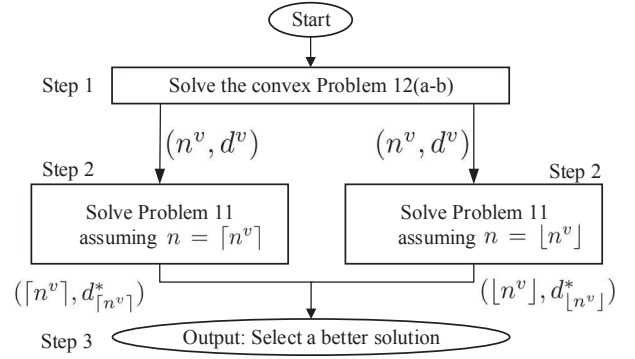}
\caption{Optimal Solution Algorithm Flowchart} \label{fig:fc}
\end{figure}


\section{Simulation Results}
In this section, we evaluate the  performance of the JoAP scheme  through simulations. We consider a 24-hour time period in all simulations. Unless specified otherwise, the charging station has $m=4$ charging ports with a charging rate $\alpha=11.5kW$. The number of parking lots is $40$. 
EVs arrive according to  a Poisson process. The parameters of arrival rate and parking time are listed in Table I according to \cite{4787536}, where the arrival rates between $4$:$01$ and $8$:$00$ (i.e., the early morning period)  are significantly lower than those of the other periods.   All EVs have the same utility parameter $\beta =0.05$ and the battery capacity $\varphi=100kWh$.\footnote{The battery specifications follow the latest information from the Tesla website: https://www.tesla.com/models.} 
For simplicity, we consider a linear waiting-time penalty $h(\omega)=c \omega$ \cite{yeo2005service}, where $c>0$ denotes the penalty rate. Our proposed JoAP algorithm is flexible enough to adapt its admission control and pricing methods to different EV arrival rates, penalty rates, and electricity prices.

\begin{center}
\captionof{table}{Simulation Parameters} \label{tab:title} 
    \begin{tabular}{| l | p{2cm}  | p{2.3cm} |}
    \hline
    Time of Day & $\lambda$ (/minutes)  & $p_e$ (\$/MWh) \\ \hline
    08:01-12:00 & 0.3  & 60 \\ \hline
    12:01-16:00 & 0.4  & 90 \\ \hline
    16:01-20:00 & 0.4  & 80 \\ \hline
    20:01-24:00 & 0.4  & 100 \\ \hline
    00:01-04:00 & 0.3  & 80 \\ \hline
    04:01-08:00 & 0.1  & 60 \\ \hline
    \end{tabular}
\end{center}

For performance comparison, we consider the following two benchmark algorithms:
\begin{enumerate}
\item Queue-length based admission (QBA):
An EV is admitted into the system only when the number of EVs already admitted is below  a threshold. For our simulations,  we set the threshold to the total number of parking lots in the charging station. Such an admission scheme has been widely used in current practice (e.g., California Plug-In Electric Vehicle Collaborative\footnote{http://www.pevcollaborative.org/workplace-charging}).  
\item Greedy admission:
An EV is admitted if  and only if doing so increases the system profit in the short-run (without considering future EV arrivals)\cite{7504120}. 
\end{enumerate}

\subsection{Average Profit Evaluation} 
  In Fig.~7, we compare the average profit per hour achieved by the three schemes under two different waiting-time penalty rates: $c=\$1/\text{min}$ and $c=\$0.4/\text{min}$. 
For each time period listed in Table I (scenarios), we simulate  1000 independent arrival processes $\mathcal{V}$ and plot the average profit performance.

We first compare the average profit of the entire day of three schemes. Fig. 7  shows that JoAP greatly outperforms the two benchmark schemes. The average profit over the whole day is 330\% and 531\% higher than that of the greedy admission scheme when the waiting-time penalty is low ($c=\$0.4/\text{min}$) and high ($c=\$1/\text{min}$), respectively. On the other hand, the widely used QBA scheme only achieves 44\% of JoAP's average profit when waiting-time penalty is low, and a negative profit when waiting-time penalty is high.

Now we investigate the performance of the three schemes in different scenarios. During low traffic period, e.g., from 4:01 to 8:00, the advantage of JoAP is not obvious. It only achieves 0.5\% and 2\% higher profit than the greedy algorithm under low and high waiting-time penalty rates, respectively. The advantage is more evident under heavy traffics, e.g., 12:01 to 24:00. Under the same traffic intensity, the advantage of JoAP over the greedy algorithm increases when $p_e$ increases. This is because the admission rate decreases rapidly when $p_e$ increases. On the other hand, the  advantage of JoAP over QBA decreases  when $p_e$ increases. This is the profit of QBA is dominated by the delay penalty, and therefore is less sensitive to the increase of electricity price $p_e$. 

It can be seen that the conventional QBA scheme performs very poorly with negative profit when the waiting-time penalty is high.  In the events of bulk arrivals, the QBA scheme admhis all EVs until there is no available parking lot and denies all the EVs that arrive later. This leads to heavy delay penalty for admitted EVs and high rejection rate for incoming EVs as well.  The greedy admission scheme has a positive but low profit  due to his inability to balance the charging schedule for the current and future EV arrivals. In fact, the greedy admission scheme always denies some EVs even under very light EV arrival traffic. In contrast,  JoAP admhis a proper number of EVs by jointly considering the EVs being served and the possible arrivals in the future, thus achieving a much higher profit than the two benchmark methods. 

\begin{figure}
\centering
   \includegraphics[width=1\linewidth]{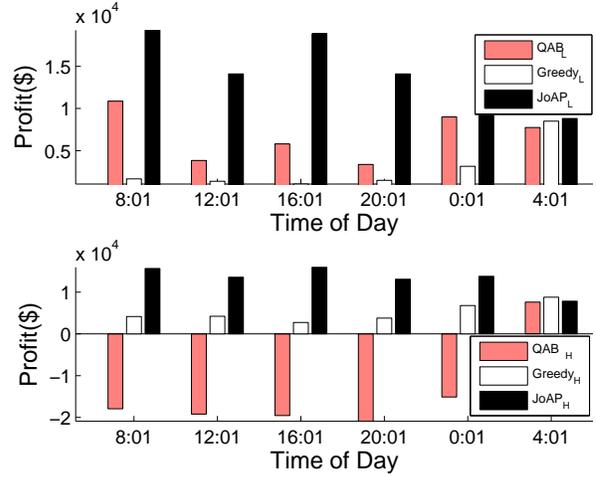}
   \label{fig:Ng2}
\caption{(a) low penalty rate ($c=\$0.4/\text{min}$) (b) high penalty rate ($c=\$1/\text{min}$)  \normalsize}
\end{figure}

\subsection{Admission Rate Evaluation}
In this subsection, we show that the average admission rate of JoAP scheme is comparable with that of the conventional QBA scheme. Fig.~8 compares the average admission rate of JoAP algorithm and the benchmarks under different penalty rates. Overall, the QBA scheme achieves the highest admission rate, i.e., 86\%, as it rejects an EV only when the parking lots are full.  However, in the some periods with moderate arrival rates, e.g., $8$:$01$ to $12$:$00$, the admission rate of the QBA scheme falls below JoAP as it is oblivious to the possible future arrivals. The overall admission rate of the greedy admission algorithm is the lowest, i.e., 70\% and 69\% in the light-penalty-rate and high-penalty rate cases, respectively.  JoAP algorithm has an admission rate 85\% and 80\% in the light-penalty-rate and high-penalty rate cases, and achieves a good balance between high admission rate  and high profit. 

\begin{figure}
\centering
   \includegraphics[width=1\linewidth]{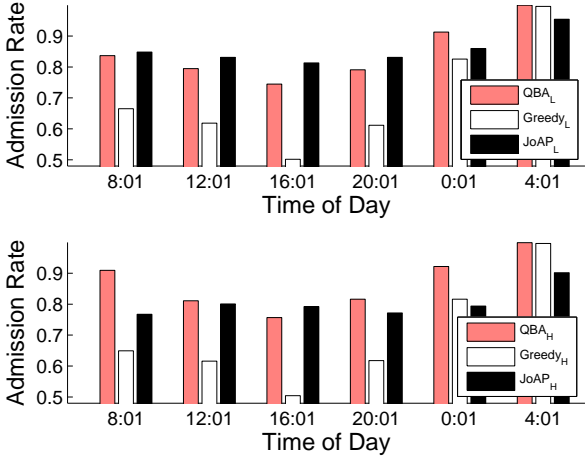}
\label{fig:wtf}
\caption{Admission rate: (a) low penalty rate (b) high penalty rate \normalsize}
\end{figure}

\subsection{Impact of $\tau$}
We have considered a fixed $\tau$ in the theoretical analysis in Section III.B. However, we have also pointed out in Section III.B that $\tau$ is can also be optimized in the JoAP admission control procedure. In Fig. 9, we numerically evaluate the performance gain if we optimize the value of $\tau$, and comparing with the case of using a fixed value of $\tau=1.01$.   For each $(c,\lambda,p_e)$, we simulate 100 independent arrival processes $\mathcal{V}$ and plot the average profit performances with  $\tau=1.01$ and the optimized  $\tau$. Averaging over all scenario, \emph{optimizing over $\tau$ increases the profit  over fixing $\tau=1.01$  only by $5.9\%$}. Therefore,  we can focus on the optimizing of $n$ (with a fixed $\tau$) in practice. 
\begin{figure}
\centering
 \includegraphics[width=\linewidth]{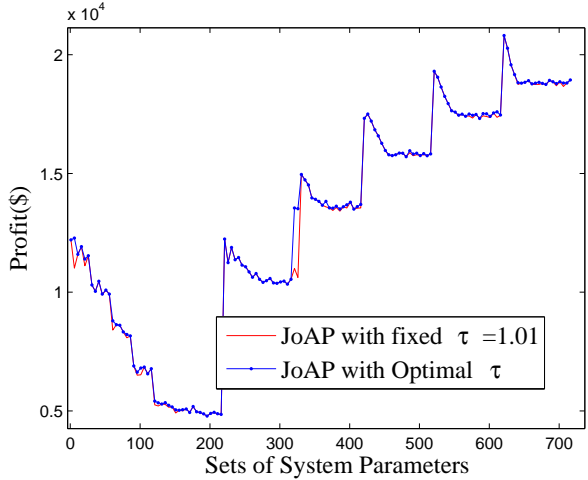}
   \caption{ Average profit with a fixed $\tau=1.01$ and the optimaized $\tau$ \normalsize}
   \label{fig:scale2}
\end{figure}
\section{Conclusions}
In this paper, we proposed a novel joint admission and pricing (JoAP) mechanism for a EV charging station to maximize his profit. In contrast to existing EV charging operation schemes, the JoAP  scheme applies a multi-sub-process admission control capable of balancing between the system admission rate and the EVs' QoS requirements according to the EV arrival rate, the electricity price, and the delay penalty. We introduced a tandem queueing model to analyze the joint admission control and scheduling process, and  proposed an efficient algorithm to compute the optimal solution.  Simulation results showed that JoAP  can effectively increase the charging station's profit while providing good QoS guarantees  to the EV users. 

In our future study, we plan to extend this work to the more general case with heterogeneous EVs. We wil further consider how  the integration of renewable and distributed energy generations  will impact the admission control and efficiency of the charging station. Thus, charging station operation under demand-sensitive electricity price due to the use of renewable energy is also an interesting future research problem.  
\appendices
\section{Proof of Lemma~1}
\begin{proof}

The mentioned two-phase-process approximation is replacing each server process with deterministic service time $d_v$ by two-phase process with exponential distribution with rate $\kappa=\frac{2}{d_v}$ Fig.~\ref{fig:scale3}. In particular, using the Laplace transform $f^{*}(s) = \frac{\kappa r_1 s + \kappa^2(r_1 + r_2 ? r_1r_2)}{ s^2 + 2\kappa s + \kappa^2(r_1 + r_2 ? r_1r_2)}$
of the density of the expiration time in the two-phase process, it is matter of simple
algebra to derive that $\kappa=\frac{2}{d_v}$, $r_1=-1$, and $r_2=\frac{5}{4}$ when the service time is deterministic and equals the constant $d_v$ \cite{tijms2008note}. 
\begin{figure}[htp]
\centering
   \includegraphics[width=0.9\linewidth]{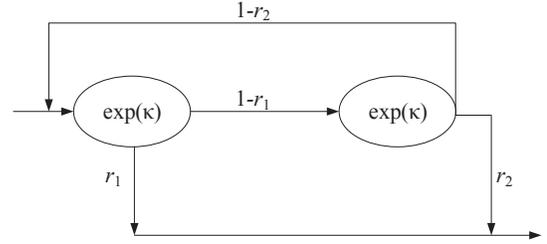}
   \caption{\small Two-phase process \normalsize}\vspace{-0.3cm}
   \label{fig:scale3}
\end{figure}

Upon this approximation, let a two-dimension pair $(s_1,s_2)$ denote the system state, where $s_1$ is the number of busy processes in phase~1, and $s_2-s_1$ is the number of busy processes in phase~2.  For particular state $(s_1,s_2)$, it can transfer to at most 6 states: $(s_1-1,s_2)$, $(s_1,s_2)$, $(s_1+1,s_2)$, $(s_1-1,s_2-1)$, $(s_1,s_2-1)$, $(s_1+1,s_2+1)$. Let $\bm{T^{k}}$ denote the generator matrix when there are total $k$ servers. Then, we calculate generator matrix $T^{k}$ in following 4 cases.
\begin{enumerate}
\item If $s_1 \geq 1$ and $s_2<k$, the system can transfer to all 6 states mentioned above. The non-zero $T^{k}$elements are,
\begin{enumerate}
\item $\bm{T^{k}}_{(s_1,s_2),(s_1-1,s_2)}= s_1 (1-r_1)\kappa$;
\item $\bm{T^{k}}_{(s_1,s_2),(s_1+1,s_2+1)}=\lambda$;
\item $\bm{T^{k}}_{(s_1,s_2),(s_1-1,s_2-1)}= s_1 r_1\kappa$;
\item $\bm{T^{k}}_{(s_1,s_2),(s_1+1,s_2)}= (s_2-s_1) (1-r_2)\kappa$;
\item $\bm{T^{k}}_{(s_1,s_2),(s_1,s_2-1)}=(s_2-s_1) r_2\kappa$;
\item $\bm{T^{k}}_{(s_1,s_2),(s_1,s_2)}=-s_2 \kappa -\lambda$.
\end{enumerate}
\item If $s_1=0$ and $s_2<k$,  the system can transfer to  $(s_1,s_2)$, $(s_1+1,s_2)$, $(s_1,s_2-1)$, $(s_1+1,s_2+1)$.
\item If $s_1 =0$ and $s_2=k$,  the system can transfer to  $(s_1,s_2)$, $(s_1+1,s_2)$, $(s_1,s_2-1)$.
\item If $s_1 =0$ and $s_2=0$, the system can transfer to  $(s_1,s_2)$, $(s_1+1,s_2+1)$.
\end{enumerate}\small

After manipulation and observation, we have,
\begin{equation}
 \mathbf{T_{k+1}}= 
\left( \begin{array}{cc}
\mathbf{T_k} & \mathbf{0}\\
\mathbf{0} & \mathbf{0}\\
\end{array}\right)+\left( \begin{array}{cc}
\mathbf{0} & \mathbf{0}\\
\mathbf{0} & \mathbf{B_{k+1}}\\
\end{array}\right)
\end{equation}\normalsize
where $ \mathbf{T_{k+1}}$ and $ \mathbf{T_{k}}$ are the generate matrix for $n=k+1$ and $n=k$, respectively. 

Let $\mathbf{x^{k+1}}=\left( \begin{array}{c}
\mathbf{y^{k+1}}\\
\mathbf{z^{k+1}}\\
\end{array}\right)$ and $\mathbf{x^{k}}$ denote the steady-state probability for $\bm{T^{K+1}}$ and $\bm{T^{K}}$, respectively. The steady-state probability of state $(s_1,s_2)$ is denoted as $x^{k}_{s_1(k-1)+s_2}$. Substitute $\mathbf{x^k}\mathbf{T_k} =0$ and $\mathbf{x^{k+1}}\mathbf{T_{k+1}} =0$ into equation (11),
we have follows,\small
\begin{subequations}
\begin{flalign}
\mathbf{y^{k+1}} &=\xi \mathbf{x^{k}}, \\
\mathbf{x_{(i,k)}^{k+1}} &=\binom{k}{i} (1-r_1)^i \mathbf{x_{(0,k)}^{k+1}}, \quad \forall i \in \{1,2,...,k\}, \\
\mathbf{x_{(i,k+1)}^{k+1}} &=\binom{k+1}{i} (1-r_1)^i \mathbf{x_{(0,k+1)}^{k+1}},\quad \forall i \in \{1,2,...,k+1\}, \\
\mathbf{x_{(0,k+1)}^{k+1}} &=\frac{d_v\lambda}{3(n+1)}\mathbf{x_{(0,k)}^{k+1}},
\end{flalign}
\end{subequations}\normalsize
where $\xi$ is a scalar. 
Let $P_k(k+1,d)=\sum_{i=0}^k \mathbf{x_{(i,k)}^{k+1}}$ denote the steady-state probability of $s_2=k$. Substitute equation (12) into $P_k(k+1,d)=\sum_{i=0}^k \mathbf{x_{(i,k)}^{k+1}}$ and $P_{k+1}(k+1,d)=\sum_{i=0}^{k+1} \mathbf{x_{(i,k+1)}^{k+1}}$. After manipulation, we get\small
\begin{equation}
P_{k+1}(k+1,d)=\frac{\tau n d \lambda }{m }P_{k}(k+1,d).
\end{equation}\normalsize
With boundary condition $\sum_{i=0}^{n} P_{i}(n,d)=1$, we get the steady-state distribution from equation (13), \small
\begin{equation}
P_i(n,d)=\frac{\frac{d^i\eta ^i}{i!}}{\sum_{j=0}^n \frac{\eta ^j}{j!}},\text{where} \quad  \eta =\frac{\tau n\lambda }{m    }.
\end{equation}\normalsize
\end{proof}
\section{Proof of Theorem~1}
\begin{proof} 
We can derive the approximated average waiting time based on Proposition~1 quoted in \cite{tijms2003first}.

\begin{proposition}
\cite{tijms2003first} For a $Ph$/$D$/$1$  queue, let $S$ and $A$ denote the service time and the inter-arrival time, respectively. The Laplace transform $a^{*}(s)=\int _0^{\infty} e^{-st}a(t)dt$ of the inter-arrival time $A$ can be written as $a^{*}(s)=\frac{a_1(s)}{a_2(s)}$, where $a(t)$ denotes the probability density function of S and $a_1(s)$ and $a_2(s)$ are two polynomials. Then, the  average waiting time can be approximated as $ \frac{\rho E(S)}{2(1-\rho)}\left[E(S^2)+E(A^2)+2E(S)\frac{a_1^{'}(0)}{a_1(0)}-2\psi\frac{a_2^{'}(0)}{a_2(0)}\right]$, where $\psi=\frac{a_2^{'}(0)-a_1^{'}(0)}{a_2(0)}$ and  $\rho$ is the load density.
\end{proposition}

We apply Proposition~1 to our tandem queue model. 
As the Lapalaze transform of $Y$ is $\mathcal{L}\{f_{Y}(x)\}=\frac{\lambda_1\lambda_2+\frac{1}{2}(\lambda_1+\lambda_2)}{(s+\lambda_1)(s+\lambda_2)}$, we have $a_1(s)=\lambda_1\lambda_2+\frac{1}{2}(\lambda_1+\lambda_2)$, $a_2(s)=(s+\lambda_1)(s+\lambda_2)$, $\psi=\frac{\frac{1}{2}\left(\lambda_1+\lambda_2\right)}{\lambda_1\lambda_2}$.
 Taking the first order derivative of $\frac{\rho }{1-\rho}$ over $d$, we have $\frac{P^{'}d+P+\frac{\lambda}{m}P^2d}{(1-\rho)^2}$, which is a positive increasing function in $d$. Substitute $\lambda_1$, $\lambda_2$ with the representation of $\mu_Y$ and $\sigma_Y$, we can express the average waiting time as $d\left[d^2+EY^2+2d\frac{\frac{1}{2}\lambda_1+\lambda_2}{\lambda_1\lambda_2}-2\frac{\frac{1}{2}\lambda_1+\lambda_2}{\lambda_1\lambda2}\frac{\lambda_1+\lambda2}{\lambda_1\lambda2}\right]=d\left[d^2+EY^2+2d \mu_Y-\mu^2_Y\right]=d\left[d^2+\sigma^2_Y+2d \mu_Y\right]$. Notice that $\mu_Y=\frac{m}{P_{\pi_n}(d)\lambda}$. Consequently, $d\left[d^2+\sigma^2_Y+2d \mu_Y\right]=d^3+2d\sigma^2_Y+\frac{md^2}{P_{\pi_n}(d)}$ is a convex increasing function in $d$ for fixed $n$.  Thus, $\omega_{\pi}(d)$ is a convex function in $d$,  which is in agreement with Kingman's formula, i.e., $w_{\pi}(d) \approx \frac{\rho d}{2(1-\rho)}\left[d^2+EY^2\right]$.

\end{proof}

\small
\bibliographystyle{ieeetr}
\bibliography{mybib}


\end{document}